\documentclass{article}
\usepackage[margin=1.5in]{geometry}
\usepackage{indentfirst}

\usepackage{cite,amsmath,amssymb,amsfonts,amsthm,enumitem}

\makeatletter
\newcommand{\startsubtheorem}{
    \def\subtheoremcounter{theorem}%
  \refstepcounter{theorem}%
  \protected@edef\theparentnumber{\csname thetheorem\endcsname}%
  \setcounter{parentnumber}{\value{theorem}}%
  \setcounter{theorem}{0}%
  \expandafter\def\csname thetheorem\endcsname{\theparentnumber.\Alph{theorem}}
}
\newcommand{\stopsubtheorem}{
    \setcounter{\subtheoremcounter}{\value{parentnumber}}%
    \expandafter\def\csname thetheorem\endcsname{\arabic{theorem}}
}
\makeatother
\newcounter{parentnumber}

\newtheorem{theorem}{Theorem}
\newtheorem{corollary}{Corollary}
\newtheorem{claim}{Claim}
\theoremstyle{remark}
\newtheorem{remark}{Remark}

\begin{document}

\title{Gabidulin Codes with Support Constrained Generator Matrices}

\author{Hikmet~Yildiz~and~Babak~Hassibi\\
Department of Electrical Engineering\\
California Institute of Technology\\
Email: \mbox{\{hyildiz, hassibi\}@caltech.edu}
}
\date{}

\maketitle

\begin{abstract}
Gabidulin codes are the first general construction of linear codes that are maximum rank distant (MRD). 
They have found applications in linear network coding, for example, when the transmitter and receiver are oblivious to the inner workings and topology of the network (the so-called incoherent regime).
The reason is that Gabidulin codes can be used to map information to linear subspaces, which in the absence of errors cannot be altered by linear operations, and in the presence of errors can be corrected if the subspace is perturbed by a small rank.
Furthermore, in distributed coding and distributed systems, one is led to the design of error correcting codes whose generator matrix must satisfy a given support constraint. 
In this paper, we give necessary and sufficient conditions on the support of the generator matrix that guarantees the existence of Gabidulin codes and general MRD codes. When the rate of the code is not very high, this is achieved with the same field size necessary for Gabidulin codes with no support constraint.
When these conditions are not satisfied, we characterize the largest possible rank distance under the support constraints and show that they can be achieved by subcodes of Gabidulin codes.
The necessary and sufficient conditions are identical to those that appear for MDS codes which were recently proven by Yildiz et al. and Lovett in the context of settling the GM-MDS conjecture.
\end{abstract}
\nocite{yildiz2018optimum,lovett2018mds} 

\section{Introduction}
Linear codes are desired to have the maximum minimum distance, for some distance measure, in order to be more resistant to errors in the channel.
If the objective is to detect and correct as many error symbols as possible, the distance measure to be used is the Hamming distance.
The Singleton bound $(n-k+1)$ is an upper bound on the largest value for the minimum Hamming distance $d_H$ a code can have, where $n$ is the length and $k$ is the dimension of the code. 
Codes achieving it are called Maximum Distance Separable (MDS) codes and a well known example for an MDS code is the Reed--Solomon code. The necessary and sufficient conditions for the existence of Reed--Solomon codes in terms of the zero structure of the generator matrix were conjectured by Dau \textit{et al.} \cite{dau2014existence}, and referred to as the GM-MDS conjecture, which was worked on by many researchers in \cite{halbawi2014distributedrs,yan2014weakly,dau2015simple,halbawi2016balanced,halbawi2016balanced2,heidarzadeh2017algebraic,song2018generalized,yildiz2018further,greaves2019reed} and finally
proved in our previous work \cite{yildiz2018optimum} and in the independent work of Lovett \cite{lovett2018mds}.

In some other scenarios, different distance metrics can be more desirable. For instance, the rank distance, $d_R$, is another metric, which can be used to design linear codes in random linear network coding or in scenarios where the transmitter and receiver are oblivious to the topology and inner workings of the network (this is often called the incoherent regime).
To see why, suppose the code is defined over an extension field $\mathbb F_{q^s}$, which can be thought of as a vector space over a base field $\mathbb F_q$, then the rank of a codeword in $\mathbb F_{q^s}^n$ is defined as the dimension of the span of the entries of the codeword over $\mathbb F_q$. Since the dimension of the span is at most the number of nonzero elements, we have $d_R\leq d_H$. Hence, a similar Singleton bound $(n-k+1)$ can be derived for the largest rank distance for a fixed code length $n$ and dimension $k$. A code achieving this is called a Maximum Rank Distance (MRD) code and Gabidulin codes due to Delsarte \cite{delsarte1978bilinear} and Gabidulin \cite{gabidulin1985theory} are the first general constructions of it. These codes require a field size of $q^s$, with $s\geq n$. Very recently, a new class of MRD codes, called twisted Gabidulin codes, have been constructed by Sheekey \cite{sheekey2015new}, which have been further generalized in \cite{lunardon2018generalized,puchinger2017further,sheekey2018new}.

In a random linear network, every node passes a random linear combination of the messages it has received to the nodes to which it is connected. In this model, the destination node will get a number of random linear combinations of the messages sent from different sources. Silva \textit{et al.} \cite{silva2008rank} showed that subspace codes or Gabidulin codes can be used to transfer messages through this network model.
In the absence of errors, the random linear combinations in the network cannot alter the transmitted subspace. In the presence of errors, or adversaries, a few nodes may transmit codewords that are not linear combinations of what they receive. This will alter the subspace by a small rank (given by the number of erroneous nodes or adversaries) and can be corrected by an MRD code.
Halbawi \textit{et al.} \cite{halbawi2014distributed} studied a scenario, where each of the source nodes has access to only a subset of all messages. They showed that subcodes of Gabidulin codes with generator matrices that have particular zero pattern (depending on what subset each source has access to) can be used under this scenario. However, they showed the existence and the code design only for networks that have up to $3$ source nodes. More specifically, they designed subcodes of Gabidulin codes with the largest rank distance under a support constraint on the generator matrix such that the rows can be divided into $3$ groups, where the rows in each group have the same zero pattern.

In this paper, we will give necessary and sufficient conditions for the existence of Gabidulin codes with support constrained generator matrices. Furthermore, if these constraints are not satisfied, we show that the largest possible rank distance can be achieved by subcodes of Gabidulin codes. Our result generalizes the result in \cite{halbawi2014distributed} to any number of source nodes in the network.
The necessary and sufficient conditions on the support constraints to guarantee the existence of Gabidulin codes and general MRD codes is identical to the conditions for MDS codes (that was recently established in \cite{yildiz2018optimum,lovett2018mds} in the context of the GM-MDS conjecture).
Furthermore, the field size is now $q^s$, with $s\geq \max \{n,k-1+\log_q k\}$. When the rate of the code is not too large ($r=\frac kn\leq 1-\frac{\log_q k-1}n$) there is no penalty in field size compared to a Gabidulin code with no support constraints.

The rest of the paper is organized as follows: 
In Section \ref{sec:coding}, after defining the rank metric and characterizing the generator matrices of Gabidulin codes, we define our problem, namely finding necessary and sufficient conditions for the existence of the Gabidulin codes with support constrained generator matrices.
Then, we solve this problem by relying on a claim (Claim \ref{claim}).
Section \ref{sec:algebraic} then proposes a purely algebraic problem on linearized polynomials that contains a more general theorem than Claim \ref{claim} and provides a detailed proof. The advantage of the generalization is that it lends itself to proof by induction.
Finally, we conclude in Section \ref{sec:conclusion}.

\section{Gabidulin Codes with Support Constraints}
\label{sec:coding}
In this section, first we will define the rank distance of a linear code, show its relation with the Hamming distance, and give its largest possible value in terms of the length $n$ and dimension $k$ of the code.
Secondly, we will write some necessary conditions on the support of the generator matrix of a code for the rank distance to achieve this largest possible value.
Thirdly, we will characterize the generator matrices of Gabidulin codes, which achieve the largest possible rank distance.
Then, we will prove that those necessary conditions are also sufficient for the existence of Gabidulin codes, which is the main result of this paper. Our proof relies on a claim (Claim \ref{claim}), which will be proven in Section \ref{sec:algebraic}, and constitutes the major technical contribution of our work.

\subsection{Rank Distance}
Let $\mathbb F_q$ be a finite field and $\mathbb F_{q^s}$ be an extension field of $\mathbb F_q$. Then, $\mathbb F_{q^s}$ forms a linear space over $\mathbb F_q$. Hence, for any $c=(c_1,\dots,c_n)\in\mathbb F_{q^s}^n$, we can define the rank of $c$ as
\begin{equation}
    \operatorname{rank}(c) = \dim(\operatorname{span}\{c_1,\dots,c_n\})
\end{equation}
Note that $\operatorname{rank}(c)$ is at most the Hamming weight of $c$, i.e. the number of nonzero entries of $c$:
\begin{equation}
\label{eq:hamming}
    \operatorname{rank}(c)\leq \|c\|_H
\end{equation}

Let $\mathcal C\subset\mathbb F_{q^s}^n$ be a linear code with $\dim\mathcal C=k$. The rank distance of $\mathcal C$ is defined as
\begin{equation}
    d_R = \min_{0\neq c\in\mathcal C} \operatorname{rank}(c)
\end{equation}
Then, by (\ref{eq:hamming}), the rank distance is less than or equal to the Hamming distance:
\begin{equation}
    d_R\leq d_H
\end{equation}
Hence, the Singleton bound on $d_H$ also holds for the rank distance: $d_R\leq n-k+1$. The codes achieving this bound are called Maximum Rank Distance (MRD) codes.
\begin{remark}
\label{remark}
    An MRD--code is also an MDS--code but the opposite is not true in general.
\end{remark}

\subsection{Support constraints (zero constraints)}
Suppose that we want to design an MRD--code under a support constraint on the generator matrix $\mathbf G\in\mathbb F_{q^s}^{k\times n}$.
We describe these support constraints through the subsets $\mathcal Z_1,\mathcal Z_2,\dots, \mathcal Z_k\subset[n]$ as follows:
\begin{equation}
\label{eq:zeroconstraints}
\forall i\in[k],\; \forall j\in \mathcal Z_i,\quad \mathbf G_{ij} = 0
\end{equation}

It is well known that \cite{dau2014existence,yildiz2018optimum,lovett2018mds} a necessary condition for a code to be MDS is
\begin{equation}
\label{eq:subsetineq}
\left|\bigcap_{i\in\Omega}\mathcal Z_i\right| + |\Omega| \leq k
\end{equation}
for all nonempty $\Omega\subset[k]$. Hence, it is also \emph{necessary} for the existence of MRD--codes by Remark \ref{remark}. Later, we will show that it is actually a \emph{sufficient} condition to design MRD--codes for fields of size $q^s$, with $s\geq \max\{n,k-1+\log_q k\}$.

Note that for $\Omega=\{i\}$, we have $|\mathcal Z_i|\leq k-1$.
In \cite[Theorem 2]{dau2014existence}, Dau \textit{et al.} showed that one can add elements from $[n]$ to each of these subsets until each has exactly $k-1$ elements by preserving (\ref{eq:subsetineq}) (We also provide a different proof in Appendix \ref{appx:hallsthm}). Note that this operation will only put more zero constraints on $\mathbf G$ but not remove any. This means that the code we design under the new constraints will also satisfy the original constraints. Therefore, without loss of generality, along with (\ref{eq:subsetineq}), we will further assume that
\begin{equation}
\label{eq:assumption}
    |\mathcal Z_i| = k-1,\quad \forall i\in[k]
\end{equation}

\subsection{Gabidulin Codes}
Gabidulin codes were introduced in \cite{delsarte1978bilinear} and \cite{gabidulin1985theory} and are the first general constructions (meaning for any $n$ and $k$) of an MRD code.
Their generator matrices are of the following form:
\begin{equation}
\mathbf G_{\text{GC}} = \begin{pmatrix}
\alpha_1^{q^0} & \alpha_2^{q^0} & \cdots & \alpha_n^{q^0}\\
\alpha_1^{q^1} & \alpha_2^{q^1} & \cdots & \alpha_n^{q^1}\\
\vdots & \vdots & & \vdots\\
\alpha_1^{q^{k-1}} & \alpha_2^{q^{k-1}} & \cdots & \alpha_n^{q^{k-1}}
\end{pmatrix}\in\mathbb F_{q^s}^{k\times n}
\end{equation}
where $\alpha_1,\alpha_2,\dots,\alpha_n\in\mathbb F_{q^s}$ are linearly independent over $\mathbb F_q$ and hence, $s\geq n$.
We remark that the linear independence of the $\alpha_i$'s over $\mathbb F_q$ is equivalent to the linear independence of any $k$ columns of $\mathbf G_{\text{GC}}$ over $\mathbb F_{q^s}$ \cite[Lemma~3.51]{lidl1997finite}.
This matrix is also known as the Moore matrix.

Furthermore, multiplying $\mathbf G_{\text{GC}}$ by an invertible matrix from the left will not change the code (i.e. the row span) but only changes the basis:
\begin{equation}
\mathbf G = \mathbf T\cdot \mathbf G_{\text{GC}}
\end{equation}
where $\mathbf T\in\mathbb F_{q^s}^{k\times k}$ is full rank.
Hence, $\mathbf G$ can be also used as a generator matrix of the same Gabidulin code.
This will allow us to introduce zeros at the desired positions on the generator matrix.

Notice that if we define the polynomials
\begin{equation}
\label{eq:polynomials}
    f_i(x)=\sum_{j=1}^k\mathbf T_{ij}x^{q^{j-1}}
\end{equation}
for $i\in[k]$, then the entries of $\mathbf G$ will be the values of the $f_i$'s evaluated at the $\alpha_j$'s i.e. $\mathbf G_{ij}=f_i(\alpha_j)$. Then, the support constraints in (\ref{eq:zeroconstraints}) on $\mathbf G$ will become root constraints on the $f_i$'s:
\begin{equation}
\label{eq:rootconstraints}
\forall i\in[k],\; \forall j\in \mathcal Z_i,\quad f_i(\alpha_j)=0
\end{equation}

In view of the above, the question we would like to ask is whether under condition (\ref{eq:subsetineq}), there exist an invertible matrix $\mathbf T$ and linearly independent $\alpha_1,\alpha_2,\dots,\alpha_n\in\mathbb F_{q^s}$ such that (\ref{eq:rootconstraints}) holds.
In other words, since $\mathbf T$ is invertible, $\mathbf G$ has the same MRD property of $\mathbf G_{\text{GC}}$, and also satisfies the support constraints in (\ref{eq:zeroconstraints}).

We should mention that a similar question for the existence of MDS codes with support constraints on the generator matrix was asked by \cite{dau2014existence} and was referred to as the GM--MDS conjecture. This was recently resolved in \cite{yildiz2018optimum,lovett2018mds}, where it was shown that under (\ref{eq:subsetineq}) MDS codes with small fields size could be constructed using Reed--Solomon codes. The current paper can be viewed as an extension of that result to rank-metric codes and Gabidulin codes.

\subsection{Example}
Let $q=2,s=4,k=3,n=4$. Suppose we have the following support constraints: $\mathcal Z_1=\{1,2\}, \mathcal Z_2=\{2,3\}, \mathcal Z_3=\{3,4\}$, i.e.,
\begin{equation}
    \mathbf G = 
    \begin{pmatrix}
    0 & 0 & \times & \times\\
    \times & 0 & 0 & \times\\
    \times & \times & 0 & 0
    \end{pmatrix}
\end{equation}
Note that these constraints satisfy (\ref{eq:subsetineq}).
We need to find $\alpha_1,\alpha_2,\alpha_3,\alpha_4\in\mathbb F_{16}$ that are linearly independent over $\mathbb F_2$ and an invertible matrix $\mathbf T\in\mathbb F_{16}^{3\times 3}$ such that

\begin{equation}
    \mathbf T\cdot \begin{pmatrix}
    \alpha_1 & \alpha_2 & \alpha_3 & \alpha_4\\
    \alpha_1^2 & \alpha_2^2 & \alpha_3^2 & \alpha_4^2\\
    \alpha_1^4 & \alpha_2^4 & \alpha_3^4 & \alpha_4^4
    \end{pmatrix} = 
    \begin{pmatrix}
    0 & 0 & \times & \times\\
    \times & 0 & 0 & \times\\
    \times & \times & 0 & 0
    \end{pmatrix}
\end{equation}
The following matrix satisfies these zero constraints (Later, we will show that this matrix is actually unique up to a scaling):
\begin{equation}
    \mathbf T = \begin{pmatrix}
    \alpha_1\alpha_2(\alpha_1+\alpha_2) & \alpha_1^2+\alpha_2^2+\alpha_1\alpha_2 & 1\\
    \alpha_2\alpha_3(\alpha_2+\alpha_3) & \alpha_2^2+\alpha_3^2+\alpha_2\alpha_3 & 1\\
    \alpha_3\alpha_4(\alpha_3+\alpha_4) & \alpha_3^2+\alpha_4^2+\alpha_3\alpha_4 & 1
    \end{pmatrix}
\end{equation}
Let's choose $\alpha_1=1,\alpha_2=a,\alpha_3=a^2,\alpha_4=a^3$ in $\mathbb F_{16}$ with the primitive polynomial $a^4+a+1$. Then, they are linearly independent over $\mathbb F_2$ and $\det\mathbf T = a^{13}\neq 0$; so, $\mathbf T$ is invertible.
Therefore,
\begin{equation}
    \mathbf G = 
    \begin{pmatrix}
    0 & 0 & a^{10} & a^3\\
    a^7 & 0 & 0 & a^{14}\\
    a^5 & a^{11} & 0 & 0
    \end{pmatrix}
\end{equation}
is the generator matrix for a Gabidulin code, which satisfies the support constraints.

Note that there are other choices of the $\alpha_i$ that can solve our problem too. However, the primary focus of this paper will be to show the existence of such a choice in general.

\subsection{Linearized Polynomials}
Polynomials in the form of (\ref{eq:polynomials}) are called \emph{linearized polynomials} ($q$-polynomials) and it is beneficial to give some of their properties before moving forward.
First, we should note that for any $a,b\in\mathbb F_{q^s}$ and $i\geq 0$, we have that $(a+b)^{q^i}=a^{q^i}+b^{q^i}$, which is commonly referred to as the \emph{Freshman's Dream} \cite{hungerford1974algebra}. Furthermore, for any $\gamma\in\mathbb F_q$, we have that $\gamma^{q^i}=\gamma$. Therefore, any linearized polynomial in the form of
\begin{equation}
\label{eq:linearizedpol}
    f(x)=\sum_{i=0}^dc_ix^{q^i},\quad c_i\in\mathbb F_{q^s}
\end{equation}
is actually a linear map $f:\mathbb F_{q^s}\to\mathbb F_{q^s}$ when $\mathbb F_{q^s}$ is considered as a linear space over $\mathbb F_q$.
Hence, the roots of $f$ form a subspace over $\mathbb F_q$.

Conversely, it can be shown that for any subspace $V\subset\mathbb F_{q^s}$, the polynomial
\begin{equation}
\label{eq:subspacedef}
    f(x)=\prod_{\beta\in V}(x-\beta)
\end{equation}
is a linearized polynomial, i.e. after expanding the product, the monomials whose exponent is not a power of $q$ will vanish \cite[Theorem~3.52]{lidl1997finite}.

The $q$-degree of the linearized polynomial $f$ in (\ref{eq:linearizedpol}) is defined as $\deg_q f=d$ if $c_d\neq 0$. Then, the $q$-degree of $f$ in (\ref{eq:subspacedef}) can be expressed as $\deg_q f=\dim V$.

We will now move on to our main problem and later revisit linearized polynomials in Section \ref{sec:algebraic}, where more properties of them will be given.

\subsection{Existence of Gabidulin Codes}
Note that by the definition in (\ref{eq:polynomials}), we have $\deg_q f_i\leq k-1$. Furthermore, since the $\alpha_j$'s are assumed to be linearly independent, by (\ref{eq:assumption}) and (\ref{eq:rootconstraints}), each $f_i$ is enforced to have $|\mathcal Z_i|=k-1$ linearly independent roots.
Therefore, the $f_1,\dots, f_k$ are uniquely defined up to a scaling, and so in monic form
\begin{equation}
\label{eq:fi_uniquely}
	f_i(x) = \prod_{\beta\in\text{span}\{\alpha_j:j\in \mathcal Z_i\}}(x-\beta),
\end{equation}
which, in turn, uniquely determines all the entries of $\mathbf T$ in terms of $\alpha_1,\dots, \alpha_n$ due to (\ref{eq:polynomials}).

Then, the problem becomes finding linearly independent $\alpha_1,\dots,\alpha_n\in\mathbb F_{q^s}$ over $\mathbb F_q$ such that $\det\mathbf T\neq 0$. In other words, we need to find $\alpha_1,\dots,\alpha_n\in\mathbb F_{q^s}$ such that
\begin{equation}
	F(\alpha_1,\dots,\alpha_n)\triangleq F_1(\alpha_1,\dots,\alpha_n)F_2(\alpha_1,\dots,\alpha_n)\neq 0
\end{equation}
where 
\begin{align}
	F_1(\alpha_1,\dots,\alpha_n) &= \det \mathbf T\\
	F_2(\alpha_1,\dots,\alpha_n) &=\begin{vmatrix}
	\alpha_1^{q^0} & \alpha_2^{q^0} & \cdots & \alpha_n^{q^0}\\
	\alpha_1^{q^1} & \alpha_2^{q^1} & \cdots & \alpha_n^{q^1}\\
	\vdots & \vdots & & \vdots\\
	\alpha_1^{q^{n-1}} & \alpha_2^{q^{n-1}} & \cdots & \alpha_n^{q^{n-1}}
	\end{vmatrix}\label{eq:f2}
\end{align}
because $\alpha_i$'s are linearly independent if and only if $F_2(\alpha_1,\dots,\alpha_n)\neq0$ \cite[Lemma~3.51]{lidl1997finite}.

It is known, by the Schwartz-Zippel Lemma, that there exist such $\alpha_j$'s in $\mathbb F_{q^s}$ if $F$ is not the zero polynomial and for all $j\in[n]$, $\deg_{\alpha_j}F < q^s$.
Note that $F_2$ is not the zero polynomial since the coefficient of the monomial $\prod_{i=1}^n\alpha_i^{q^{i-1}}$ in $F_2$ is $1$ because it can only be obtained through multiplication of the diagonals. Furthermore, if Claim \ref{claim} below is true, we can conclude that $F$ is not the zero polynomial.
\begin{claim}
\label{claim}
$\det\mathbf T$ is not the zero polynomial if (\ref{eq:subsetineq}) is satisfied.\hfill$\diamond$
\end{claim}
We will give the proof of Claim \ref{claim} later in Section \ref{sec:algebraic} by proving a slightly more general statement. Therefore, in this section, we will proceed by assuming that it is true.
Then, $F$ is not the zero polynomial and the only question that remains is ``what is the largest value of $\deg_{\alpha_j} F$ over all $j\in [n]$?'', whose answer, in turn, can be used as a sufficient lower bound on the size of the extension field where such $\alpha_j$'s exist.

Notice from (\ref{eq:f2}) that for a fixed $\alpha_j$, the degree of $F_2$ as a polynomial in $\alpha_j$ is
$$\deg_{\alpha_j} F_2=q^{n-1}$$
Now, we will compute $\deg_{\alpha_j} F_1$.
From (\ref{eq:polynomials}), recall that for any $i,\ell\in[k]$, $\mathbf T_{i\ell}$ is the coefficient of $x^{q^{\ell-1}}$ in $f_i(x)$. Since $f_i(x)$ is monic, $\mathbf T_{ik}=1$. For $\ell < k$, $\mathbf T_{i\ell}$ is a polynomial in $\alpha_j$ and $\deg_{\alpha_j}\mathbf T_{i\ell}\leq \deg_{\alpha_j} f_i(x)$ (When writing $\deg_{\alpha_j} f_i(x)$, we consider $f_i(x)$ as a polynomial in $\alpha_j$).

To find $\deg_{\alpha_j} f_i$, consider the definition of $f_i$ in (\ref{eq:fi_uniquely}). Suppose that $j\in \mathcal Z_i$ (Otherwise, $\deg_{\alpha_j}f_i=0$). Let $\mathcal Z'_i=\mathcal Z_i-\{j\}$ and define $f'_i$ as
\begin{equation}
	f'_i(x) = \prod_{\beta\in\text{span}\{\alpha_{j'}:j'\in \mathcal Z'_i\}}(x-\beta)
\end{equation}
which is a linearized polynomial with $\deg_q f'_i=|\mathcal Z'_i|=k-2$ and hence as a usual polynomial $\deg_x f'_i(x) = q^{k-2}$. Since $j\notin\mathcal Z'_i$, $f'_i(x)$ is independent of $\alpha_j$; therefore, we can also write $\deg_{\alpha_j}f'_i(\alpha_j)=q^{k-2}$. Furthermore, we can write that
\begin{align}
	f_i(x)
	&= \prod_{\beta\in\text{span}\{\alpha_{j'}:j'\in \mathcal Z_i\}}(x-\beta)\\
	&= \prod_{\gamma\in\mathbb F_q}\prod_{\beta\in\text{span}\{\alpha_{j'}:j'\in \mathcal Z'_i\}}(x-\gamma\alpha_j-\beta)\\
	&= \prod_{\gamma\in\mathbb F_q} f'_i(x-\gamma\alpha_j)\\
	&= \prod_{\gamma\in\mathbb F_q} (f'_i(x)-\gamma f'_i(\alpha_j))\\
	&= (f'_i(x))^q-(f'_i(\alpha_j))^{q-1}f'_i(x)
\end{align}
where the last step is because of the identity $\prod_{\gamma\in\mathbb F_q}(x-a\gamma)=x^q-a^{q-1}x$.

Hence, $\deg_{\alpha_j}\mathbf T_{i\ell}\leq\deg_{\alpha_j}f_i(x) \leq (q-1)\deg_{\alpha_j} f'_i(\alpha_j)=(q-1)q^{k-2}$. Then,
\begin{align}
    \deg_{\alpha_j} F_1
    &= \deg_{\alpha_j}\det\mathbf T\\
    &\leq \max_{\sigma\in S_k}\sum_{\ell=1}^k\mathbf \deg_{\alpha_j}\mathbf T_{\sigma(\ell),\ell}\\
    &\leq (k-1)(q-1)q^{k-2}
\end{align}
where $S_k$ denotes the set of permutations of $[k]$ and in the last inequality, recall that $\mathbf T_{ik}=1$, whose degree is $0$.
As a result,
\begin{equation}
    \deg_{\alpha_j} F
    \leq q^{n-1}+(k-1)(q-1)q^{k-2}
\end{equation}

So, if the field size is larger than this bound, i.e. $q^s>q^{n-1}+(k-1)(q-1)q^{k-2}$, then there exist $\alpha_1,\dots,\alpha_n\in\mathbb F_{q^s}$ such that $F(\alpha_1,\dots,\alpha_n)\neq 0$.
As a result, we have the following theorem. Note that if $s\geq n$ and $s\geq k-1+\log_qk$, then
\begin{equation}
    q^s=q^{s-1}+(q-1)q^{s-1}\geq q^{n-1}+(q-1)kq^{k-2}>q^{n-1}+(k-1)(q-1)q^{k-2}
\end{equation}
\begin{theorem}
    \label{thm:gabidulin}
	For any $s\geq\max\{n, k-1+\log_qk\}$, if (\ref{eq:subsetineq}) is satisfied, then there exists a Gabidulin code in $\mathbb F_{q^s}$ of length $n$ and dimension $k$ such that its generator matrix satisfies the support constraints in (\ref{eq:zeroconstraints}).\hfill$\diamond$
\end{theorem}

\subsection{Subcodes of Gabidulin codes}
If the necessary and sufficient condition in (\ref{eq:subsetineq}) is not satisfied, we cannot have an MDS code with the prescribed support constraints, and by fiat we cannot have an MRD code or a Gabidulin code. However, we can still ask whether a code with the largest possible rank distance can be achieved. In fact, we can show that the largest rank distance can be achieved by subcodes of Gabidulin codes for a large enough field sizes.
In \cite{yildiz2018optimum}, the following upper bound on the Hamming distance is noted:
\begin{equation}
    d_H\leq n-\ell+1
\end{equation}
where
\begin{equation}
        \ell\triangleq\max_{\emptyset\neq\Omega\subset[k]} \left(\left|\bigcap_{i\in\Omega}\mathcal Z_i\right| + |\Omega|\right)\geq k
\end{equation}
Since the rank distance of the code is upper bounded by the Hamming distance, we have that
\begin{equation}
\label{eq:drupperbound}
    d_R\leq n-\ell+1
\end{equation}

\begin{theorem}
    \label{thm:subcode}
    Suppose $s\geq \max\{n,\ell-1+\log_q\ell\}$. Then, there exists a subcode of a Gabidulin code in $\mathbb F_{q^s}$ with length $n$, dimension $k$, and rank distance $d_R=n-\ell+1$ such that its generator matrix satisfies (\ref{eq:zeroconstraints}).\hfill$\diamond$
\end{theorem}
\begin{proof}
    Define $\mathcal Z_{k+1}=\cdots=\mathcal Z_{\ell}=\emptyset$. Then, for any nonempty $\Omega\subset[\ell]$,
    \begin{equation}
        \left|\bigcap_{i\in\Omega}\mathcal Z_i\right| + |\Omega|\leq\ell
    \end{equation}
    Hence, by Theorem \ref{thm:gabidulin}, there exists a Gabidulin code of dimension $\ell$ with an $\ell\times n$ generator matrix $\mathbf G$ having zeros dictated by $\mathcal Z_1,\dots, \mathcal Z_{\ell}$. Since it is an MRD--code, its rank distance is $n-\ell+1$. The first $k$ rows of $\mathbf G$ will generate a subcode whose rank distance $d_R$ is as good as the Gabidulin code:  $d_R\geq n-\ell+1$. Hence, this subcode achieves the largest possible rank distance given in (\ref{eq:drupperbound}).
\end{proof}

\section{Proof of Claim \ref{claim} (and More)}
\label{sec:algebraic}
In this section, first we will extend the definition of linearized polynomials by allowing their coefficients to be multivariate polynomials.
Then, we will propose a more general statement than Claim \ref{claim}, namely Theorems \ref{thm:linearizedpol} and \ref{thm:matrixlang}, which, in fact, arise when trying to apply a proof by induction to Claim \ref{claim}.
Our generalization will be written in two different forms. Theorem \ref{thm:linearizedpol} will be in terms of linearized polynomials, whereas Theorem \ref{thm:matrixlang} will be in terms of matrices. However, both are equivalent and more general than Claim \ref{claim}. We will give a sketch of the proof in the language of matrices while the detailed proof will be given in the language of polynomials. We should emphasize that the material presented here in the matrix language is only for a better illustration of Theorem \ref{thm:linearizedpol}.

\subsection{Problem Setup}
Consider a finite field $\mathbb F_q$ and an extension field $\mathsf R_0=\mathbb F_{q^s}$. For $n\geq 1$, let $\mathsf R_n\triangleq\mathbb F_{q^s}[x_1,\dots, x_n]$ be the ring of multivariate polynomials in the indeterminates $x_1,x_2\dots,x_n$ over $\mathbb F_{q^s}$.

Recall that the notation $\mathsf R_n[x]$ denotes the ring of polynomials in the indeterminate $x$, whose coefficients are drawn from $\mathsf R_n$ (the coefficients are multivariate polynomials in $x_1,\dots,x_n$), i.e.,
\begin{equation}
    \mathsf R_n[x]\triangleq \left\{\sum_{i=0}^d c_ix^i\middle| d\geq 0, c_0,\dots,c_d\in\mathsf R_n\right\}
\end{equation}

The set of linearized polynomials over $\mathsf R_n$ is a subset of $\mathsf R_n[x]$, which we define as:
\begin{equation}
	\mathsf L_n\triangleq \left\{ \sum_{i=0}^d c_ix^{q^i}\middle| d\geq 0, c_0,\dots,c_d\in\mathsf R_n \right\}\subset\mathsf R_n[x]
\end{equation}

The $q$-degree of $f\in\mathsf L_n$ is defined as $\deg_q f=d$ if $f= \sum_{i=0}^d c_ix^{q^i}$ and  $c_d\neq 0$. We also take $\deg_q 0 = -\infty$.
Since $\mathsf L_n\subset \mathsf R_n[x]$, for any $f,g\in\mathsf L_n$, we will continue to use $\gcd\{f,g\}$ and $f\mid g$ notations by treating as $f,g\in\mathsf R_n[x]$.\\

We note the following properties of $\mathsf L_n$ (See \cite[Chapter 3]{lidl1997finite} as a reference textbook, where these properties are proven for $\mathsf L_0$, i.e., when the coefficients of the linearized polynomials are from $\mathbb F_{q^s}$. The same proofs can be extended to $\mathsf L_n$. We also give the proofs of \ref{p:ring} and \ref{p:subspace} in Appendix \ref{appx:properties} as the other properties are obvious):

\begin{enumerate}[label=\textbf{P\arabic*.},ref=P\arabic*]
	\item \label{p:ring}
	$\mathsf L_n$ is a ring with no zero divisors under the addition and the composition operation $\circ$.
	
	\item \label{p:deg}
	For any $f,g\in\mathsf L_n$, $\deg_q (f\circ g) = \deg_q (f) + \deg_q (g)$.
	
	\item \label{p:subspace}
	For any finite-dimensional subspace $V\subset\mathsf R_n$ over $\mathbb F_q$ and $t\geq 0$,
	\begin{equation}
	f=\prod_{\beta\in V}(x-\beta)^{q^t}\in\mathsf L_n
	\end{equation}
	and $\deg_q f=t+\dim V$.
	
	\item \label{p:xdividesf}
	For any $f\in\mathsf L_n$, if $x^{q^t}\mid f$, then $\exists f'\in\mathsf L_n$ such that  $f=f'\circ x^{q^t}$.\\
	
	\item \label{p:xdividesfog}
	For any $f,g\in\mathsf L_n$, if $x^{q^t}\mid f$, then $x^{q^t}\mid f\circ g$ and $x^{q^t}\mid g\circ f$.\\
	
	\item \label{p:xdividesgofnotf}
	For any $f,g\in\mathsf L_n$, if $x^q\nmid f$ and $x^{q^t}\mid g\circ f$, then $x^{q^t}\mid g$.\\
\end{enumerate}

We are interested in linearized polynomials of the following form:
\begin{equation}
    \mathsf f(\mathcal Z, t) \triangleq \prod_{\beta\in\text{span}\{x_i:i\in\mathcal Z\}}(x-\beta)^{q^t}\in\mathsf L_n,
    \qquad t\geq 0, \mathcal Z\subset[n]
\end{equation}
Note that these are linearized polynomials in light of \ref{p:subspace} above.
Furthermore, since the $x_i$'s are assumed to be indeterminates, any nontrivial linear combination of them is nonzero, i.e. the $x_i$'s are linearly independent. Hence,
\begin{equation}
\label{eq:degree}
    \deg_q\mathsf f(\mathcal Z,t) = t+\dim(\operatorname{span}\{x_i:i\in\mathcal Z\}) = t+|\mathcal Z|
\end{equation}
For $k\geq 1$, define the set of linearized polynomials in this form with $q$-degree at most $k-1$:
\begin{equation}
	\mathcal L_{n,k}
	\triangleq \left\{\mathsf f(\mathcal Z,t) \middle| t\geq 0,\mathcal Z\subset[n]
	\mbox{ s.t. } t+|\mathcal Z|\leq k-1 \right\} \subset \mathsf L_n
\end{equation}

We also note the following properties with regard to $\mathcal L_{n,k}$, whose proofs appear in Appendix \ref{appx:properties}.
\begin{enumerate}[resume*]
	\item \label{p:gcd}
	For any $f_1=\mathsf f(\mathcal Z_1,t_1),f_2=\mathsf f(\mathcal Z_2,t_2)\in\mathcal L_{n,k}$, we have
	$$\gcd \{f_1,f_2\}=\mathsf f(\mathcal Z_1\cap\mathcal Z_2,\min\{t_1,t_2\}) \in\mathcal L_{n,k}$$
	
	\item \label{p:f2dividesf1}
	For any $f_1,f_2\in\mathcal L_{n,k}$, if $f_2\mid f_1$, then $\exists f'_1\in\mathsf L_n,\: f_1=f'_1\circ f_2$.
	
	\item \label{p:substitution}
	Let $f=\mathsf f(\mathcal Z,t)\in\mathcal L_{n,k}$ and let $f'=f|_{x_n=0}\in\mathsf L_{n-1}$ (substitute $x_n=0$ in each coefficient of $f$). Then, $f'\in\mathcal L_{n-1,k}$ and
	\begin{equation}
	    f'=\begin{cases}
	    \mathsf f(\mathcal Z,t) & n\notin\mathcal Z\\
	    \mathsf f(\mathcal Z-\{n\},t+1) & n\in\mathcal Z
	    \end{cases}
	\end{equation}
\end{enumerate}

As a final note, it will be insightful to describe the composition operation between linearized polynomials in matrix language. It is known that multiplying two polynomials is equivalent to multiplying two Toeplitz matrices since both perform the convolution operation.
Now, we will give the analog when composing two linearized polynomials. Let $f = \sum_{i=0}^dc_ix^{q^i}\in\mathsf L_n$. For $b-a\geq d$, we define the following matrix:
\begin{equation*}
	\mathbf S_{a\times b}(f) = \begin{pmatrix}
		c_0^{q^0} & c_1^{q^0} & \cdots & c_{b-a}^{q^0}\\
		& c_0^{q^1} & c_1^{q^1} & \cdots & c_{b-a}^{q^1}\\
		&& \ddots & \ddots & & \ddots\\
		&&& c_0^{q^{a-1}} & c_1^{q^{a-1}} &\cdots& c_{b-a}^{q^{a-1}}
	\end{pmatrix}
\end{equation*}
where $c_i=0$ for $i>d$. Note that $a$ and $b$ are parameters that define the dimensions of the matrix $\mathbf S_{a\times b}(f)$, which is why we subscript $\mathbf S$ by $a\times b$.
For any linearized polynomials $f_1,f_2\in\mathsf L_n$, we have that
\begin{equation}
\label{eq:prodmatrixs}
	\mathbf S_{a\times b}(f_1\circ f_2) = \mathbf S_{a\times c}(f_1)\cdot \mathbf S_{c\times b}(f_2)
\end{equation}
for any $a,b,c$ such that $c-a\geq \deg_q f_1$ and $b-c\geq \deg_q f_2$. The proof follows straightforward calculations by definition. As a special case, when $f_1=x^{q^t}$, $f_2=f=\sum_{i=0}^dc_ix^{q^i}$, and $f_1\circ f_2=f^{q^t}$, we can write for $b-a\geq d$,
\begin{align}
    \kern-.4em
    \mathbf S_{a\times (b+t)}(f^{q^t})
    &= \mathbf S_{a\times (a+t)}(x^{q^t})\cdot\mathbf S_{(a+t)\times (b+t)}(f)\\
    &= \begin{pmatrix}\mathbf 0_{a\times t}&\mathbf I_{a\times a}\end{pmatrix}
    \cdot\mathbf S_{(a+t)\times (b+t)}(f)\\\label{eq:firstcolumns}
    &= \begin{pmatrix}
		0&\cdots & 0 & c_0^{q^t} & c_1^{q^t} & \cdots & c_{b-a}^{q^t}\\
		0&\cdots & 0 & & c_0^{q^{t+1}} & c_1^{q^{t+1}} & \cdots & c_{b-a}^{q^{t+1}}\\
		\vdots& & \vdots & && \ddots & \ddots & & \ddots\\
		\makebox[0pt][l]{$\smash{\underbrace{\phantom{\begin{matrix}0&\cdots&0\end{matrix}}}_{t}}$}0&\cdots & 0 & &&& c_0^{q^{t+a-1}} & c_1^{q^{t+a-1}} &\cdots& c_{b-a}^{q^{t+a-1}}
	\end{pmatrix}\\\nonumber
\end{align}
Since by definition, $\mathsf f(\mathcal Z,t)={f'}^{q^t}$ for some $f'\in\mathsf L_n$, we have the following property.
\begin{enumerate}[resume*]
    \item \label{p:firstcolumns}
    Let $f=\mathsf f(\mathcal Z,t)$ and $r\geq 0$.
    Then the first $r+t$ columns of $\mathbf S_{a\times (b+r)}(f^{q^r})$ are all zero.
\end{enumerate}

\subsection{Main Result}
Theorem \ref{thm:linearizedpol} is a more general statement than Claim \ref{claim} given in Section \ref{sec:coding}
and it is the analog of \cite[Theorem 3]{yildiz2018optimum} for linearized polynomials.

\startsubtheorem
\begin{theorem}
	\label{thm:linearizedpol}
	Let $k\geq m\geq 1$ and $n\geq 0$. Then, for any $f_1,f_2,\dots, f_m\in\mathcal L_{n,k}$, the following are equivalent:
	\begin{enumerate}[label=(\roman*)]
		\item For all $g_1,g_2,\dots,g_m\in\mathsf L_n$ and $r\geq 0$ such that ${\deg_q(g_i\circ f_i) \leq k-1}$, we have
		\begin{equation}
		\sum_{i=1}^m g_i\circ x^{q^r}\circ f_i = 0 \implies g_1=g_2=\cdots=g_m=0
		\end{equation}
		
		\item For all nonempty $\Omega\subset[m]$, we have
		\begin{equation}
		\label{eq:gcdineq}
		k-\deg_q\gcd_{i\in\Omega}f_i \geq \sum_{i\in\Omega} (k-\deg_q f_i)
		\end{equation}
		\hfill$\diamond$
	\end{enumerate}
\end{theorem}

Before moving to the proof, in order to see how Claim \ref{claim} becomes a special case of Theorem \ref{thm:linearizedpol}, we will give an equivalent way of writing it in terms of matrices with entries from $\mathsf R_n$. This will also allow us to see its connection with \cite[Theorem 3]{yildiz2018optimum}.

For $i\in[m]$, let $f_i=\mathsf f(\mathcal Z_i,t_i)\in\mathcal L_{n,k}$ (i.e. $\mathcal Z_i\subset[n], t_i\geq 0$ such that $|\mathcal Z_i|+t_i\leq k-1$).
For $r\geq 0$, we will write $\mathbf S(f_i^{q^r})$ instead of $\mathbf S_{(k-t_i-|\mathcal Z_i|)\times (k+r)}(f_i^{q^r})$ for the ease of notation.
By \ref{p:firstcolumns}, $\mathbf S(f_i^{q^r})$ will look like as follows, where the $\times$'s represent the nonzero entries:
\begin{align}
\mathbf S(f_i^{q^r})
&=\begin{pmatrix}
        0 & \cdots & 0 & \times & \times & \cdots & \times\\
        0 & \cdots & 0 &  & \times & \times & \cdots & \times\\
        \vdots && \vdots &&&\ddots & \ddots & &\ddots\\
        \makebox[0pt][l]{$\smash{\underbrace{\phantom{\begin{matrix}0&\cdots&0\end{matrix}}}_{r+t_i}}$}0 & \cdots & 0&
        \makebox[0pt][l]{$\smash{\underbrace{\phantom{\begin{matrix}\times&\times&\cdots\end{matrix}}}_{k-1-t_i-|\mathcal Z_i|}}$}
        &&&
        \makebox[0pt][l]{$\smash{\underbrace{\phantom{\begin{matrix}\times&\times&\cdots&\times\end{matrix}}}_{|\mathcal Z_i|+1}}$}
        \times & \times & \cdots & \times\\
    \end{pmatrix}\!\!\!\!\left.\vphantom{\begin{pmatrix}0\\0\\\vdots\\0\end{pmatrix}}
				\right\} {\scriptstyle k-t_i-|\mathcal Z_i|}\\
\end{align}

Then, applying (\ref{eq:prodmatrixs}) to the expression $g_i\circ x^{q^r}\circ f_i=g_i\circ f_i^{q^r}$ in Theorem \ref{thm:linearizedpol} yields
\begin{equation}
    \mathbf S_{1\times (k+r)}(g_i\circ x^{q^r}\circ f_i)
    =\mathbf u_i\cdot
    \mathbf S(f_i^{q^r})
\end{equation}
where $\mathbf u_i=\mathbf S_{1\times (k-t_i-|\mathcal Z_i|)}(g_i)$ is a row vector.
Therefore, we can write
\begin{equation}
    \mathbf S_{1\times (k+r)}\left(\sum_{i=1}^mg_i\circ x^{q^r}\circ f_i\right)
    =\begin{pmatrix}\mathbf u_1&\cdots&\mathbf u_m\end{pmatrix}\cdot
    \begin{pmatrix}
	\mathbf S(f_1^{q^r})\\
	\vdots\\
	\mathbf S(f_m^{q^r})
	\end{pmatrix}
\end{equation}
which is a linear combination of the rows of
\begin{equation}\label{eq:matrix_m}
	\mathbf M(r) = \begin{pmatrix}
	\mathbf S(f_1^{q^r})\\
	\vdots\\
	\mathbf S(f_m^{q^r})
	\end{pmatrix}_{\sum_{i=1}^m(k-t_i-|\mathcal Z_i|)\times (k+r)}
\end{equation}
Hence, ($i$) in Theorem \ref{thm:linearizedpol} is equivalent to saying the matrix $\mathbf M(r)$ has full row rank.
Note that the first $r$ columns of $\mathbf M(r)$ are zero since the first $r+t_i$ columns of $\mathbf S(f_i^{q^r})$ are so.

Furthermore, ($ii$) in Theorem \ref{thm:linearizedpol} can be written in terms of the $\mathcal Z_i$'s and the $t_i$'s in lights of (\ref{eq:degree}) and \ref{p:gcd}.
Therefore, Theorem \ref{thm:linearizedpol} is equivalent to Theorem \ref{thm:matrixlang} below.
\begin{theorem}
	\label{thm:matrixlang}
	For $i\in[m]$, let $\mathcal Z_i\subset[n], t_i\geq 0$ such that $|\mathcal Z_i|+t_i\leq k-1$.
	Then, the matrix $\mathbf M(r)$ defined in (\ref{eq:matrix_m}) has full row rank for all $r\geq 0$ if and only if for all nonempty $\Omega\subset[m]$,
	\begin{equation}
		\label{eq:ineq_matrixthm}
		k-\left|\bigcap_{i\in\Omega}\mathcal Z_i\right|-\min_{i\in\Omega}t_i\geq \sum_{i\in\Omega} (k-t_i-|\mathcal Z_i|)
	\end{equation}
	\hfill$\diamond$
\end{theorem}
\stopsubtheorem

As a special case, when $m=k$, $|\mathcal Z_i|=k-1$, $t_i=0$, and $r=0$, each block in $\mathbf M(r)$ becomes a row vector with coefficients of $f_i=\mathsf f(\mathcal Z_i,t_i)=\sum_{i=1}^kc_{ij}x^{q^{j-1}}$:
$$\mathbf S_{(k-t_i-|\mathcal Z_i|)\times(k+r)}(f_i^{q^0})=\mathbf S_{1\times(k+r)}(f_i)=\begin{pmatrix}c_{i1}&c_{i2}&\cdots&c_{ik}\end{pmatrix}$$
Hence, we have Corollary \ref{corollary} below, which is Claim \ref{claim} in Section \ref{sec:coding}.
\begin{corollary}
	\label{corollary}
	For $i\in[k]$, let $\mathcal Z_i\subset[n]$ with $|\mathcal Z_i|=k-1$. Then,
	$$k\geq \left|\bigcap_{i\in\Omega}\mathcal Z_i\right| + |\Omega|,\qquad\forall\:\emptyset\neq\Omega\subset[k]$$
	if and only if
	$$\det\begin{pmatrix}
	c_{11} & c_{12} & \dots & c_{1k}\\
	c_{21} & c_{22} & \dots & c_{2k}\\
	\vdots & \vdots & & \vdots\\
	c_{k1} & c_{k2} & \dots & c_{kk}
	\end{pmatrix}\neq 0
	$$
	where $c_{ij}$'s are defined as the coefficients of $f_i=\mathsf f(\mathcal Z_i,0)=\sum_{i=1}^kc_{ij}x^{q^{j-1}}$.
	\hfill$\diamond$
\end{corollary}

\subsection{Sketch of the proof of Theorem \ref{thm:matrixlang}}
The proof given here for Theorem \ref{thm:matrixlang} omits certain steps that the interested reader can fill in. The complete proof of the equivalent Theorem \ref{thm:linearizedpol} is given in Section \ref{sec:proof} and includes each and every step.

The following identity (\ref{eq:matrix_identity}) will be very useful throughout the proof.

For any $\Omega\subset[m]$ (wlog assume $\Omega=\{1,2,\dots,\ell\}$), we have $f_i=f'_i\circ f_0$ for $i\in[\ell]$, where $f_0=\gcd_{i\in\Omega}f_i$. Then, we can write (with the appropriate dimensions for $\mathbf S(\,\cdot\,)$)
\begin{equation}
	\label{eq:matrix_identity}
	\begin{pmatrix}
	 \mathbf S(f_1^{q^r})\\
	 \vdots\\
	 \mathbf S(f_{\ell}^{q^r})
	\end{pmatrix}
	=\underbrace{\begin{pmatrix}
		\mathbf S({f'_1}^{q^r})\\
		\vdots\\
		\mathbf S({f'_{\ell}}^{q^r})
	\end{pmatrix}}_{\begin{bmatrix}
	\mathbf 0_{*\times r} & \mathbf B'
	\end{bmatrix}}
	 \cdot \underbrace{\mathbf S(f_0)}_{\begin{bmatrix}
	\times \\ \mathbf S(f_0^{q^r})
	\end{bmatrix}} = \mathbf B'\cdot \mathbf S(f_0^{q^r})
\end{equation}
where the matrix $\mathbf B'$ has $(k-\left|\bigcap_{i\in\Omega}\mathcal Z_i\right|-\min_{i\in\Omega}t_i)$ columns and $\sum_{i\in\Omega} (k-t_i-|\mathcal Z_i|)$ rows. Note that these are respectively the left and right hand sides in (\ref{eq:ineq_matrixthm}).

Therefore, if (\ref{eq:ineq_matrixthm}) does not hold then $\mathbf B'$ will be a tall matrix and will not have full row rank, which solves $\implies$ direction. For the other direction, we will try to reduce the problem to the one that has a smaller $k,m$, or $n$ in order to do an inductive proof. We look into two cases:
\begin{enumerate}[label=Case \arabic*.,leftmargin=2.5\parindent]
	\item (\ref{eq:ineq_matrixthm}) is tight for some $2\leq |\Omega|\leq m-1$.
	\item (\ref{eq:ineq_matrixthm}) is strict for all $2\leq |\Omega|\leq m-1$.
\end{enumerate}

In the first case, the matrix $\mathbf B'$ becomes a square matrix. Hence,
\begin{align}
	\begin{pmatrix}
	\mathbf S(f_1^{q^r})\\
	\vdots\\
	\mathbf S(f_m^{q^r})
	\end{pmatrix}
	&=\begin{pmatrix}
	\mathbf B'\mathbf S(f_0^{q^r})\\
	\mathbf S(f_{\ell+1}^{q^r})\\
	\vdots\\
	\mathbf S(f_m^{q^r})
	\end{pmatrix}\\\label{eq:decomposition}
	&=\begin{pmatrix}
	\mathbf B'\\
	&\mathbf I\\
	&&\ddots\\
	&&&\mathbf I
	\end{pmatrix}
	\begin{pmatrix}
	\mathbf S(f_0^{q^r})\\
	\mathbf S(f_{\ell+1}^{q^r})\\
	\vdots\\
	\mathbf S(f_m^{q^r})
	\end{pmatrix}
\end{align}
This will reduce the problem into two smaller problems: The first one is showing that the matrix on the right in (\ref{eq:decomposition}) has full row rank. The second one is showing that $\mathbf B'$ is non-singular or that $\mathbf B'\cdot \mathbf S(f_0^{q^r})$, which is equal to the first $\ell$ blocks (see (\ref{eq:matrix_identity})), has full row rank. Both are smaller problems (in terms of the number of blocks) and one can show that both satisfy the inequalities in (\ref{eq:ineq_matrixthm}).

In the second case, since the inequalities are strict except for $|\Omega|=1,m$, we have some flexibility to play with the sets.
For example, we can remove an element $j$ from all the sets $\mathcal Z_i$'s containing $j$ and increase $t_i$ by $1$ (This corresponds to Case 2c in the proof of Theorem \ref{thm:linearizedpol}). This operation sets $x_j=0$ in the matrix $\mathbf M(r)$ and we can claim that if $\mathbf M(r)|_{x_j=0}$ has full row rank, then so does $\mathbf M(r)$. Hence, it reduces $n$ in the problem to $n-1$.
Furthermore, it can be shown that except for two corner cases (see Case 2a and 2b), one can carefully choose such an element $j$ so that removing it from the sets will not break (\ref{eq:ineq_matrixthm}) for $|\Omega|=m$.

The only two corner cases are when none or only one of the $t_i$'s is zero. If $t_i\geq 1$ for all $i\in[m]$ (i.e. the first $r+1$ columns of $\mathbf M(r)$ are all zero), then decreasing $k$ and each $t_i$ by $1$ and increasing $r$ by $1$ will reduce the problem into a smaller one (see Case 2a). If there is a unique zero, say $t_1=0$ (see Case 2b), then the first $r+1$ columns of $\mathbf S(f_i^{q^r})$ will be zero only for $i\geq 2$. Then, the matrix will look like
\begin{equation}
	\mathbf M(r) = \begin{pmatrix}
	0&\cdots&0&\times & \times & \cdots&\times \\
	0&\cdots&0&0&\times & \times & \cdots&\times \\
	\vdots&&\vdots&\vdots&&\ddots &\ddots &&\ddots \\
	0&\cdots&0&0&&&\times &\times &\cdots & \times\\\hline
	0&\cdots&0&0&\times&\times & \cdots & \times\\
	\vdots&&\vdots&\vdots&&\ddots & && \ddots \\
	0&\cdots&0&0&&&\times&\times &\cdots & \times \\\hline
	&&&&&\vdots
	\end{pmatrix}
\end{equation}
Hence, the first row is definitely not in the span of the other rows because it contains a nonzero in the $(r+1)$th column while the others do not. So, we can reduce the problem by removing the first row. This will decrease $k$ and every $t_i$ except $t_1$ by $1$ (and maybe $m$ too if there is a single row in the first block). Again, it can be shown that this operation does not violate (\ref{eq:ineq_matrixthm}).

\subsection{Proof of Theorem \ref{thm:linearizedpol}}
\label{sec:proof}
Let $f_i=\mathsf f(\mathcal Z_i,t_i)$.
For the ease of notation we will write $\textstyle f_{\Omega} \triangleq \gcd_{i\in\Omega}f_i$, which, by \ref{p:gcd}, is equal to
\begin{equation}\textstyle
    f_{\Omega} = \mathsf f\left(\bigcap\limits_{i\in\Omega}\mathcal Z_i,\: \min\limits_{i\in\Omega}t_i\right)
\end{equation}
We will first show the trivial direction ($(i)\implies (ii)$), then do induction for the other direction ($(ii)\implies (i)$).\\

\noindent$(i)\implies (ii)$:\\

Suppose that $(ii)$ does not hold and wlog, assume that for $\Omega=\{1,2,\dots,\ell \}$,
$$k-\deg_q f_{\Omega}<\sum_{i\in\Omega}(k-\deg_q f_i)$$
For $i\in\Omega$, let $f_i=f'_i\circ f_{\Omega}$ for some $f'_i\in\mathsf L_n$ (see \ref{p:f2dividesf1}). Then, for $r=0$ and for $g_1,\dots,g_{\ell}\in\mathsf L_n$ such that $\deg_q(g_i\circ f_i)\leq k-1$, in $(i)$, the equation $\sum_{i\in\Omega}g_i\circ f'_i=0$  defines homogeneous linear equations in coefficients of $g_i$'s. The number of variables is $\sum_{i\in\Omega}k-\deg_q f_i$ and the number of equations is  at most $k-\deg_q f_{\Omega}$. So, one can find $g_1,\dots,g_{\ell}$, not all zero, that solves this linear system.\\

\noindent$(ii)\implies(i)$:\\

We will do induction on parameters $(k,m,n)$ considered in the lexicographical order.

For $(k,m=1,n)$, $(i)$ always holds due to \ref{p:ring}: $g_1\circ x^{q^r}\circ f_1= 0\implies g_1=0$.

For $(k,m\geq 2,n=0)$, $(ii)$ never holds: $n=0\implies f_i=x^{q^{t_i}}$ for some $t_i$ for every $i$. Suppose $t_1\leq t_2$, then for $\Omega=\{1,2\}$, (\ref{eq:gcdineq}) becomes $k-t_1\geq(k-t_1)+(k-t_2)$, which contradicts with $|\mathcal Z_i|+t_i\leq k-1$.

For $k\geq m\geq 2$ and $n\geq 1$ assume that the statement ($(ii)\implies(i)$) is true for parameters $(k',m',n')<(k,m,n)$.
Take any $f_1,\dots, f_m\in\mathcal L_{n,k}$ for which, $(ii)$ is true. We will prove that $(i)$ holds under the following cases:

\begin{enumerate}[label=Case \arabic*.,leftmargin=2.5\parindent]
	\item $\exists\, \Omega\subset[m]$ with $2\leq |\Omega|\leq m-1$ such that (\ref{eq:gcdineq}) holds with equality.
	\item $\forall\, \Omega\subset[m]$ with $2\leq |\Omega|\leq m-1$, (\ref{eq:gcdineq}) holds strictly and any of these three:
	\begin{enumerate}[label=Case 2\alph*.]
	    \item For all $i\in[m]$, $t_i\geq 1$.
	    \item There exists a unique $i\in[m]$ such that $t_i=0$.
	    \item There exist at least two zero $t_i$.
	\end{enumerate}
\end{enumerate}

We will reduce $m$ in Case 1, $k$ in Case 2a and 2b, and $n$ in Case 2c. Note that since $k\geq m$, reducing $k$ sometimes may also reduce $m$, which may happen in Case 2b but will not happen in Case 2a, where we show $k\geq m+1$.\\

\textbf{Case 1:}
Wlog, assume that for $\Omega'=\{1,2,\dots,\ell \}$, 
$$k-\deg_q f_0 = \sum_{i\in\Omega'} (k-\deg_q f_i)$$
where $f_0=f_{\Omega'}$.
By \ref{p:f2dividesf1}, for $i\in[\ell]$, there exists $f'_i\in\mathsf L_n$ such that $f_i=f'_i\circ f_0$.

We will look at two smaller problems:
$(f_1,\dots,f_{\ell})\in\mathcal L_{n,k}^{\ell}$ and $(f_0,f_{\ell+1},\dots,f_m)\in\mathcal L_{n,k}^{m-\ell+1}$. Since $\ell<m$ and $m-\ell+1 < m$, the statement holds for both by the induction hypothesis.\\
It is trivial that $(ii)$ holds for $(f_1,\dots, f_{\ell})$ and for $(f_0,f_{\ell+1},\dots,f_m)$ when $0\notin\Omega$.
We will show that it also holds for 
$(f_0,f_{\ell+1},\dots,f_m)$ when $0\in\Omega$:
\begin{align}
k-\deg_q f_{\Omega}
&= k-\deg_q \gcd\{f_0, f_{(\Omega-\{0\})}\}\\
&= k-\deg_q \gcd\{f_{\Omega'},f_{(\Omega-\{0\})} \}\\
&\leq \sum_{i\in\Omega'\cup (\Omega-\{0\})}(k-\deg_q f_i)\\
&= \sum_{i\in\Omega'}(k-\deg_q f_i) + \sum_{i\in(\Omega-\{0\})}(k-\deg_q f_i)\\
&=(k-\deg_q f_0) + \sum_{i\in(\Omega-\{0\})}(k-\deg_q f_i)\\
&=\sum_{i\in\Omega}(k-\deg_q f_i)
\end{align}
Hence, by the induction hypothesis, $(i)$ holds for both $(f_1,\dots, f_{\ell})$ and $(f_0,f_{\ell+1},\dots,f_m)$.
Now, we will show that it also holds for $(f_1,\dots,f_m)$:

Suppose that for some $r\geq 0$ and $g_1,\dots, g_m\in\mathsf L_n$ with $\deg_q g_i\circ f_i\leq k-1$ for $i\in[m]$, we have
$$\sum_{i=1}^mg_i\circ x^{q^r}\circ f_i = 0$$
Since $x^{q^r}\mid \sum_{i=1}^{\ell}g_i\circ x^{q^r}\circ f'_i$, by \ref{p:xdividesf}, we can write
$$\sum_{i=1}^{\ell}g_i\circ x^{q^r}\circ f'_i=g_0\circ x^{q^r}$$
for some $g_0\in\mathsf L_n$.
Then,
\begin{align*}
    0
    &= \sum_{i=1}^m g_i\circ x^{q^r} \circ f_i\\
    &= \sum_{i=1}^{\ell} g_i\circ x^{q^r} \circ f'_i\circ f_0 + \sum_{i=\ell+1}^m g_i\circ x^{q^r} \circ f_i\\
    &= g_0\circ x^{q^r}\circ f_0 + \sum_{i=\ell+1}^m g_i\circ x^{q^r} \circ f_i
\end{align*}
Hence, $g_0=g_{\ell+1}=\cdots=g_m=0$. Then,
\begin{equation}
g_0\circ x^{q^r}\circ f_0 = \sum_{i=1}^{\ell}g_i\circ x^{q^r}\circ f_i = 0
\end{equation}
Hence, $g_1=\cdots=g_{\ell}=0$. Then, all $g_i$'s are zero.\\

\textbf{Case 2a:} For all $i\in[m]$, $f_i=x^q\circ f'_i$, where $f'_i=\mathsf f(\mathcal Z_i,t_i-1)\in\mathcal L_{n,k-1}$. Note that since $\min_{i\in[m]}t_i\geq 1$, we have $\deg_q f_{[m]}\geq 1$ and for $\Omega=[m]$, $(ii)$ implies
$$k-1\geq k-\deg_q f_{[m]}\geq \sum_{i\in[m]}(k-\deg_q f_i)\geq m$$
By the induction hypothesis, the statement is true for $(f'_1,\dots,f'_m)$ with parameters $(k-1,m,n)$.

$(ii)$ holds for $(f'_1,\dots,f'_m)$ because for any nonempty $\Omega\subset[m]$,
\begin{align*}
    k-1-\deg_q f'_{\Omega}
    &= k-\deg_q f_{\Omega}\\
    &\geq \sum_{i\in\Omega} (k-\deg_q f_i)\\
    &= \sum_{i\in\Omega} (k-1-\deg_q f'_i)
\end{align*}
Hence, $(i)$ holds for $(f'_1,\dots,f'_m)$ too and we will show that it also holds for $(f_1,\dots, f_m)$:

Suppose that for some $r\geq 0$ and $g_1,\dots, g_m\in\mathsf L_n$ with $\deg_q g_i\circ f_i\leq k-1$ for $i\in[m]$, we have
$$\sum_{i=1}^mg_i\circ x^{q^r}\circ f_i = 0$$
Then,
\begin{align*}
    0
    &= \sum_{i=1}^m g_i\circ x^{q^r}\circ f_i\\
    &= \sum_{i=1}^m g_i\circ x^{q^r}\circ x^q\circ f'_i\\
    &= \sum_{i=1}^m g_i\circ x^{q^{r+1}}\circ f'_i
\end{align*}
Hence, $g_1=\cdots=g_m=0$.\\

\textbf{Case 2b:} Suppose that $t_m=0$ and for $i\in[m-1]$, $t_i\geq 1$. For $i\in[m-1]$, let $f_i=x^q\circ f_i'$, where $f'_i=\mathsf f(\mathcal Z_i,t_i-1)\in\mathcal L_{n,k-1}$ and let $f'_m=f_m\in\mathcal L_{n,k}$.
Note that $f'_m\in\mathcal L_{n,k-1}$ if and only if $\deg_q f'_m \leq k-2$, in which case for $\Omega=[m]$, $(ii)$ implies
$$k\geq k-\deg_q f_{[m]}\geq \sum_{i\in[m]}(k-\deg_q f_i)\geq m+1$$

By the induction hypothesis, the statement is true for $(f'_1,\dots,f'_m)$ with parameters $(k-1,m,n)$ if $k\geq m+1$ (or $\deg_q f'_m\leq k-2$) and for $(f'_1,\dots,f'_{m-1})$ with parameters $(k-1,m-1,n)$.

We will show that $(ii)$ holds for $(f'_1,\dots,f'_m)$ when $k$ is replaced by $k-1$.
If $m\notin\Omega$, it is similar to Case 2a.
For $m\in\Omega$, first observe that since each root of $f_m$ has a multiplicity of $1$, we have $\gcd\{f_m,f'_i\} = \gcd\{f_m,f_i\}$ for $i\in[m-1]$; hence, $f_{\Omega}=f'_{\Omega}$.
Then,
\begin{align*}
    (k-1)-\deg_q f'_{\Omega}
    &= -1 + k - \deg_q f_{\Omega}\\
    &\geq -1+\sum_{i\in\Omega} (k-\deg_q f_i)\\
    &=(k-1-\deg_q f_m) + \sum_{i\in\Omega-\{m\}} (k-\deg_q f_i)\\
    &=(k-1-\deg_q f'_m) + \sum_{i\in\Omega-\{m\}} (k-1-\deg_q f'_i)\\
    &=\sum_{i\in\Omega}(k-1-\deg_q f'_i)
\end{align*}
Hence, $(i)$ also holds for $f'_i$'s.

Suppose that for some $r\geq 0$ and $g_1,\dots, g_m\in\mathsf L_n$ with $\deg_q(g_i\circ f_i)\leq k-1$, we have
$$\sum_{i=1}^m g_i\circ x^{q^r} \circ f_i=0$$
Then,
\begin{align*}
0
&=\sum_{i=1}^m g_i\circ x^{q^r}\circ f_i\\
&=g_m\circ x^{q^r}\circ f_m + \sum_{i=1}^{m-1}g_i\circ x^{q^{r}}\circ x^q\circ f'_i\\
&= g_m\circ x^{q^r}\circ f_m + \underbrace{\sum_{i=1}^{m-1}g_i\circ x^{q^{r+1}}\circ f'_i}_{\text{divisible by } x^{q^{r+1}} \text{ due to \ref{p:xdividesfog}}}
\end{align*}
Hence, $g_m\circ x^{q^r}\circ f_m$ is divisible by $x^{q^{r+1}}$ and since $x^q\nmid f_m$ (because $t_m=0$),
by \ref{p:xdividesgofnotf}, $x^{q^{r+1}}\mid g_m\circ x^{q^r}$. Then, by \ref{p:xdividesf}, we can write $g_m=g'_m\circ x^q$ for some $g'_m\in\mathsf L_n$ with $\deg_q g'_i=\deg_q g_i-1$.

If $\deg_q f_m=k-1$, then, $\deg_q g'_m\leq-1$, which implies $g_m=0$. Then, $g_1,\dots,g_{m-1}$ are also zero since $(i)$ holds for $(f'_1,\dots,f'_{m-1})$ with parameters $(k-1,m-1,n)$.

If $\deg_q f_m\leq k-2$, then,

\begin{equation}
0=g'_m\circ x^{q^{r+1}}\circ f'_m + \sum_{i=1}^{m-1}g_i\circ x^{q^{r+1}}\circ f'_i
\end{equation}
Then, $g_1=\cdots=g_{m-1}=g'_m=0$ since $(i)$ holds for $(f'_1,\dots,f'_m)$ with parameters $(k-1,m,n)$. Then all $g_i$'s are zero.\\

\textbf{Case 2c:} Wlog, assume that $t_{m-1}=t_m=0$.
If $\mathcal Z_{m-1}=\mathcal Z_m$, then for $\Omega=\{m-1,m\}$, $(ii)$ implies
$$k-\deg_q f_m=k-\deg_q \gcd\{f_{m-1},f_m\}\geq (k-\deg_q f_{m-1})+(k-\deg_q f_m) $$
which contradicts with $\deg_q f_{m-1}\leq k-1$.
Hence, either $\mathcal Z_{m-1}\neq [n]$ or $\mathcal Z_m\neq [n]$. Wlog, assume $\mathcal Z_m\neq [n]$ and $n\notin\mathcal Z_m$.

Now, we will substitute $x_n=0$. 
Let $f'_i=f_i\mid_{x_n=0}$.
By \ref{p:substitution}, $f'_i\in\mathcal L_{n-1,k}$ and
\begin{equation}
    f'_i=\mathsf f(\mathcal Z'_i,t'_i)=\begin{cases}
    \mathsf f(\mathcal Z_i,t_i) & n\notin\mathcal Z_i\\
    \mathsf f(\mathcal Z_i-\{n\},t_i+1) & n\in\mathcal Z_i
    \end{cases}
\end{equation}

By the induction hypothesis the statement is true for $(f'_1,\dots,f'_m)$ with parameters $(k,m,n-1)$. We will show that it satisfies $(ii)$:

For $|\Omega|=1$, it is trivial.

For $2\leq|\Omega|\leq m-1$, then
\begin{align}
    k-\deg_q f'_{\Omega}
    &= k-\left|\bigcap_{i\in\Omega}\mathcal Z'_i\right|-\min_{i\in\Omega}t'_i\\
    &\leq k-\left(\left|\bigcap_{i\in\Omega}\mathcal Z_i\right|-1\right)-\min_{i\in\Omega}t_i\\\label{eq:plusone}
    &= k+1-\deg_q f_{\Omega}\\
    &\leq \sum_{i\in\Omega}(k-\deg_q f_i)\\
    &= \sum_{i\in\Omega}(k-\deg_q f'_i)
\end{align}
where the last inequality is because we assume (\ref{eq:gcdineq}) holds strictly for $2\leq|\Omega|\leq m-1$ and the first inequality is because $t'_i\geq t_i$ and
$$\left|\bigcap_{i\in\Omega}\mathcal Z'_i\right| = \left|\bigcap_{i\in\Omega}\mathcal Z_i-\{n\}\right|\geq \left|\bigcap_{i\in\Omega}\mathcal Z_i\right|-1$$

For $|\Omega|=m$, (\ref{eq:gcdineq}) was not strict; however, there is no need to have the $+1$ in (\ref{eq:plusone}) since
$$n\notin \mathcal Z_m\implies n\notin\bigcap_{i\in[m]}\mathcal Z_i\implies \left|\bigcap_{i\in\Omega}\mathcal Z'_i\right|=\left|\bigcap_{i\in\Omega}\mathcal Z_i\right|$$

Therefore, $(ii)$ holds for $f'_i$'s. Hence, so does $(i)$.

Suppose that for some $g_1,\dots, g_m\in\mathsf L_n$, not all zero, with $\deg_q(g_i\circ f_i)\leq k-1$, we have
$$\sum_{i=1}^m g_i\circ x^{q^r}\circ f_i=0$$
We can further assume that at least one coefficient of one $g_i$ is not divisible by $x_n$. (Otherwise, divide them by $x_n$).
Define $g'_i=g_i\mid_{x_n=0}\in\mathsf L_{n-1}$. Then, the $g'_i$'s are not all zero. We can write
\begin{equation}
\sum_{i=1}^mg'_i\circ x^{q^r}\circ f'_i=\left.\left(\sum_{i=1}^mg_i\circ x^{q^r}\circ f_i\right)\right|_{x_n=0}=0\mid_{x_n=0}=0
\end{equation}
Then, $g'_1=\cdots=g'_m=0$. Contradiction.
\qed

\newpage
\section{Conclusion}
\label{sec:conclusion}
In this paper, we extended our proof technique in \cite{yildiz2018optimum} for Reed--Solomon codes to Gabidulin codes by writing an analog of the algebraic-combinatorial problem presented there. The main challenge in extending the result to Gabidulin codes was that, unlike polynomial multiplication, the composition operation between linearized polynomials is not commutative.
As a result, we showed that the work of Halbawi \textit{et al.} \cite{halbawi2014distributed} can be applied to networks with any number of source nodes, which had been shown only for $3$ source nodes.

Theorem \ref{thm:gabidulin} only claims the existence of Gabidulin codes since its proof is based on the multivariate polynomial $F(\alpha_1,\dots,\alpha_n)$ being not identically zero.
The same observation applies to subcodes of Gabidulin codes.
In order to explicitly construct a Gabidulin code, we need to explicitly specify the evaluations points $\alpha_1,\dots,\alpha_n$ for which $F$ takes a nonzero value.
One possible algorithm could be to generate random evaluation points until $F$ takes a nonzero value. However, currently, we do not know the average complexity of this algorithm.
Hence, how to construct such codes efficiently remains an important open problem.
As a special case, when the generator matrix is systematic (i.e. $\mathcal Z_i=[k]\backslash \{i\}$), constructions of Gabidulin codes are given in \cite{neri2018systematic}.

\bibliographystyle{IEEEtran} 
\bibliography{refs} 

\begin{thebibliography}{10}
\providecommand{\url}[1]{#1}
\csname url@samestyle\endcsname
\providecommand{\newblock}{\relax}
\providecommand{\bibinfo}[2]{#2}
\providecommand{\BIBentrySTDinterwordspacing}{\spaceskip=0pt\relax}
\providecommand{\BIBentryALTinterwordstretchfactor}{4}
\providecommand{\BIBentryALTinterwordspacing}{\spaceskip=\fontdimen2\font plus
\BIBentryALTinterwordstretchfactor\fontdimen3\font minus
  \fontdimen4\font\relax}
\providecommand{\BIBforeignlanguage}[2]{{%
\expandafter\ifx\csname l@#1\endcsname\relax
\typeout{** WARNING: IEEEtran.bst: No hyphenation pattern has been}%
\typeout{** loaded for the language `#1'. Using the pattern for}%
\typeout{** the default language instead.}%
\else
\language=\csname l@#1\endcsname
\fi
#2}}
\providecommand{\BIBdecl}{\relax}
\BIBdecl

\bibitem{yildiz2018optimum}
H.~Yildiz and B.~Hassibi, ``Optimum linear codes with support constraints over
  small fields,'' in \emph{2018 IEEE Information Theory Workshop (ITW)}.\hskip
  1em plus 0.5em minus 0.4em\relax IEEE, 2018, pp. 1--5.

\bibitem{lovett2018mds}
S.~Lovett, ``{MDS} matrices over small fields: A proof of the {GM-MDS}
  conjecture,'' in \emph{2018 IEEE 59th Annual Symposium on Foundations of
  Computer Science (FOCS)}.\hskip 1em plus 0.5em minus 0.4em\relax IEEE, 2018,
  pp. 194--199.

\bibitem{dau2014existence}
S.~H. Dau, W.~Song, and C.~Yuen, ``On the existence of {MDS} codes over small
  fields with constrained generator matrices,'' in \emph{2014 IEEE
  International Symposium on Information Theory}.\hskip 1em plus 0.5em minus
  0.4em\relax IEEE, 2014, pp. 1787--1791.

\bibitem{halbawi2014distributedrs}
W.~Halbawi, T.~Ho, H.~Yao, and I.~Duursma, ``Distributed {Reed--Solomon} codes
  for simple multiple access networks,'' in \emph{2014 IEEE International
  Symposium on Information Theory}.\hskip 1em plus 0.5em minus 0.4em\relax
  IEEE, 2014, pp. 651--655.

\bibitem{yan2014weakly}
M.~Yan, A.~Sprintson, and I.~Zelenko, ``Weakly secure data exchange with
  generalized {Reed--Solomon} codes,'' in \emph{2014 IEEE International
  Symposium on Information Theory}.\hskip 1em plus 0.5em minus 0.4em\relax
  IEEE, 2014, pp. 1366--1370.

\bibitem{dau2015simple}
S.~H. Dau, W.~Song, and C.~Yuen, ``On simple multiple access networks,''
  \emph{IEEE Journal on Selected Areas in Communications}, vol.~33, no.~2, pp.
  236--249, 2015.

\bibitem{halbawi2016balanced}
W.~Halbawi, Z.~Liu, and B.~Hassibi, ``Balanced {Reed--Solomon} codes,'' in
  \emph{2016 IEEE International Symposium on Information Theory (ISIT)}.\hskip
  1em plus 0.5em minus 0.4em\relax IEEE, 2016, pp. 935--939.

\bibitem{halbawi2016balanced2}
------, ``Balanced {Reed--Solomon} codes for all parameters,'' in \emph{2016
  IEEE Information Theory Workshop (ITW)}.\hskip 1em plus 0.5em minus
  0.4em\relax IEEE, 2016, pp. 409--413.

\bibitem{heidarzadeh2017algebraic}
A.~Heidarzadeh and A.~Sprintson, ``An algebraic-combinatorial proof technique
  for the {GM-MDS} conjecture,'' in \emph{2017 IEEE International Symposium on
  Information Theory (ISIT)}.\hskip 1em plus 0.5em minus 0.4em\relax IEEE,
  2017, pp. 11--15.

\bibitem{song2018generalized}
W.~Song and K.~Cai, ``Generalized {Reed--Solomon} codes with sparsest and
  balanced generator matrices,'' in \emph{2018 IEEE International Symposium on
  Information Theory (ISIT)}.\hskip 1em plus 0.5em minus 0.4em\relax IEEE,
  2018, pp. 1--5.

\bibitem{yildiz2018further}
H.~Yildiz and B.~Hassibi, ``Further progress on the {GM-MDS} conjecture for
  {Reed--Solomon} codes,'' in \emph{2018 IEEE International Symposium on
  Information Theory (ISIT)}.\hskip 1em plus 0.5em minus 0.4em\relax IEEE,
  2018, pp. 16--20.

\bibitem{greaves2019reed}
G.~Greaves and J.~Syatriadi, ``{Reed--Solomon} codes over small fields with
  constrained generator matrices,'' \emph{IEEE Transactions on Information
  Theory}, 2019.

\bibitem{delsarte1978bilinear}
P.~Delsarte, ``Bilinear forms over a finite field, with applications to coding
  theory,'' \emph{Journal of Combinatorial Theory, Series A}, vol.~25, no.~3,
  pp. 226--241, 1978.

\bibitem{gabidulin1985theory}
E.~M. Gabidulin, ``Theory of codes with maximum rank distance,'' \emph{Problemy
  Peredachi Informatsii}, vol.~21, no.~1, pp. 3--16, 1985.

\bibitem{sheekey2015new}
J.~Sheekey, ``A new family of linear maximum rank distance codes,'' \emph{arXiv
  preprint arXiv:1504.01581}, 2015.

\bibitem{lunardon2018generalized}
G.~Lunardon, R.~Trombetti, and Y.~Zhou, ``Generalized twisted {Gabidulin}
  codes,'' \emph{Journal of Combinatorial Theory, Series A}, vol. 159, pp.
  79--106, 2018.

\bibitem{puchinger2017further}
S.~Puchinger, J.~Sheekey \emph{et~al.}, ``Further generalisations of twisted
  {Gabidulin} codes,'' \emph{arXiv preprint arXiv:1703.08093}, 2017.

\bibitem{sheekey2018new}
J.~Sheekey, ``New semifields and new {MRD} codes from skew polynomial rings,''
  \emph{arXiv preprint arXiv:1806.00251}, 2018.

\bibitem{silva2008rank}
D.~Silva, F.~R. Kschischang, and R.~Koetter, ``A rank-metric approach to error
  control in random network coding,'' \emph{IEEE Transactions on Information
  Theory}, vol.~54, no.~9, pp. 3951--3967, 2008.

\bibitem{halbawi2014distributed}
W.~Halbawi, T.~Ho, and I.~Duursma, ``Distributed {Gabidulin} codes for
  multiple-source network error correction,'' in \emph{2014 International
  Symposium on Network Coding (NetCod)}.\hskip 1em plus 0.5em minus 0.4em\relax
  IEEE, 2014, pp. 1--6.

\bibitem{lidl1997finite}
R.~Lidl and H.~Niederreiter, \emph{Finite fields}.\hskip 1em plus 0.5em minus
  0.4em\relax Cambridge University Press, 1997.

\bibitem{hungerford1974algebra}
T.~W. Hungerford, \emph{Algebra}.\hskip 1em plus 0.5em minus 0.4em\relax
  Springer, 1974.

\bibitem{neri2018systematic}
A.~Neri, ``Systematic encoders for generalized {Gabidulin} codes and the
  $q$-analogue of {Cauchy} matrices,'' \emph{arXiv preprint arXiv:1805.06706},
  2018.

\end{thebibliography}

\newpage
\appendix
\section{Proofs of some properties of linearized polynomials}
\label{appx:properties}
\begin{itemize}
	\item[\textbf{\ref{p:ring}.}] \emph{$\mathsf L_n$ is a ring with no zero divisors under the addition and the composition operation $\circ$.}
	\begin{proof}
    	Note that for any $a,b\in\mathsf R_n[x]$,
    	\begin{equation}
    	\label{eq:freshmansdream}
    	    (a+b)^q=a^q+b^q
    	\end{equation}
    	Let $f=\sum_{i=0}^{d_1}f_ix^{q^i},\: g=\sum_{i=0}^{d_2}g_ix^{q^i} \in\mathsf L_n$. Then,
    	\begin{align*}
    	    f\circ g
    	    &=f\left(\sum_{i=0}^{d_2}g_ix^{q^i}\right)\\
    	    &=\sum_{i=0}^{d_2}f(g_ix^{q^i})\\
    	    &=\sum_{i=0}^{d_2}\sum_{j=0}^{d_1}f_jg_i^{q^j}x^{q^{i+j}}\in\mathsf L_n
    	\end{align*}
    	Furthermore, if $f,g\neq 0$, then $f\circ g\neq 0$ since the leading coefficient, $f_{d_1}g_{d_2}^{q^{d_1}}$ is nonzero.
    	Hence, $\mathsf L_n$ has no zero divisors.
    	
    	By (\ref{eq:freshmansdream}), for any $f,g,h\in\mathsf L_n$,
    	\begin{align*}
    	  f\circ (g+h)
    	  &= f(g(x)+h(x))\\
    	  &=f(g(x))+f(h(x))\\
    	  &=f\circ g+f\circ h  
    	\end{align*}
    	The other ring properties are trivial.
	\end{proof}
	
	\item[\textbf{\ref{p:subspace}.}] \emph{For any finite-dimensional subspace $V\subset\mathsf R_n$ over $\mathbb F_q$ and $t\geq 0$,
	\begin{equation}
	f=\prod_{\beta\in V}(x-\beta)^{q^t}\in\mathsf L_n
	\end{equation}
	and $\deg_q f=t+\dim V$}
	\begin{proof}
	    It is sufficient to prove it for $t=0$ because
    	$$\prod_{\beta\in V}(x-\beta)^{q^t}=x^{q^t}\circ \prod_{\beta\in V}(x-\beta)$$
    	We do induction on $\dim V$.
    	If $\dim V=1$, then $V=\{\alpha a: \alpha\in\mathbb F_q\}$ for some $a\in\mathsf R_n$ and
    	\begin{align*}
    	    \prod_{\beta\in V}(x-\beta)
    	    &= \prod_{\alpha\in\mathbb F_q}(x-\alpha a)\\
    	    &= x^q-a^{q-1}x\in\mathsf L_n
    	\end{align*}
    	Suppose $V'\subset V$ is a subspace such that $\dim V' = \dim V-1$ and suppose $f'=\prod_{\beta\in V'}(x-\beta)\in\mathsf L_n$. Then, $V=\{\alpha a + b: \alpha\in \mathbb F_q, b\in V'\}$ for some $a\in\mathsf R_n$ and
    	\begin{align*}
    	    \prod_{\beta\in V}(x-\beta)
    	    &= \prod_{\alpha\in\mathbb F_q, b\in V'}(x-\alpha a-b)\\
    	    &= \prod_{\alpha\in\mathbb F_q}\prod_{b\in V'}((x-\alpha a)-b)\\
    	    &= \prod_{\alpha\in\mathbb F_q}f'(x-\alpha a)\\
    	    &= \prod_{\alpha\in\mathbb F_q}(f'(x)-\alpha f'(a))\\
    	    &= [x^q-(f'(a))^{q-1}x]\circ f'\in\mathsf L_n
    	\end{align*}
	\end{proof}

	\item[\textbf{\ref{p:gcd}.}] \emph{For any $f_1=\mathsf f(\mathcal Z_1,t_1),f_2=\mathsf f(\mathcal Z_2,t_2)\in\mathcal L_{n,k}$, we have $$\gcd \{f_1,f_2\}=\mathsf f(\mathcal Z_1\cap\mathcal Z_2,\min\{t_1,t_2\}) \in\mathcal L_{n,k}$$\vspace{-2em}}
	\begin{proof}
	    Note that each root of $f_i$ has a multiplicity of $q^{t_i}$. Therefore, the roots of $\gcd$ of $f_1$ and $f_2$ will be the elements of
	    $$\operatorname{span}\{x_j: j\in \mathcal Z_1\}\cap\operatorname{span}\{x_j: j\in \mathcal Z_2\}=\operatorname{span}\{x_j: j\in \mathcal Z_1\cap \mathcal Z_2\},$$
	    each with a multiplicity of $\min\{t_1,t_2\}$.
	\end{proof}
	
	\item[\textbf{\ref{p:f2dividesf1}.}] \emph{If $f_1,f_2\in\mathcal L_{n,k}$ and $f_2\mid f_1$, then $\exists f'_1\in\mathsf L_n,\: f_1=f'_1\circ f_2$.}
	\begin{proof}
    	Let $f_1=\mathsf f(\mathcal Z_1,t_1)$ and $f_2=\mathsf f(\mathcal Z_2,t_2)$. Since each root of $f_i$ has a multiplicity of $q^{t_i}$, we have $t_2\leq t_1$. Furthermore, the roots of $f_2$ are also roots of $f_1$:
    	$$\operatorname{span}\{x_j:j\in\mathcal Z_2\}\subset\operatorname{span}\{x_j:j\in\mathcal Z_1\}$$
    	Hence, $\mathcal Z_2\subset \mathcal Z_1$. Then,
    	\begin{align*}
    	    f_1
    	    &= \prod_{\beta\in\operatorname{span}\{x_j:j\in\mathcal Z_1\}} (x-\beta)^{q^{t_1}}\\
    	    &= \prod_{a\in\operatorname{span}\{x_j:j\in\mathcal Z_1-\mathcal Z_2\}}\:\prod_{b\in\operatorname{span}\{x_j:j\in\mathcal Z_2\}} (x-a-b)^{q^{t_1}}\\
    	    &= \prod_{a\in\operatorname{span}\{x_j:j\in\mathcal Z_1-\mathcal Z_2\}} (f_2(x-a))^{q^{t_1-t_2}}\\
    	    &= \prod_{a\in\operatorname{span}\{x_j:j\in\mathcal Z_1-\mathcal Z_2\}} (f_2(x)-f_2(a))^{q^{t_1-t_2}}\\
    	    &= \prod_{\beta\in\operatorname{span}\{f_2(x_j):j\in\mathcal Z_1-\mathcal Z_2\}} (f_2(x)-\beta)^{q^{t_1-t_2}}\\
    	    &= f'\circ f_2
    	\end{align*}
    	where $\displaystyle f'=\prod_{\beta\in\operatorname{span}\{f_2(x_j):j\in\mathcal Z_1-\mathcal Z_2\}}(x-\beta)^{q^{t_1-t_2}}\in\mathsf L_n$.
	\end{proof}
	
	\item[\textbf{\ref{p:substitution}.}] \emph{Let $f=\mathsf f(\mathcal Z,t)\in\mathcal L_{n,k}$ and let $f'=f|_{x_n=0}\in\mathsf L_{n-1}$ (substitute $x_n=0$ in each coefficient of $f$). Then, $f'\in\mathcal L_{n-1,k}$ and
	\begin{equation}
	    f'=\begin{cases}
	    \mathsf f(\mathcal Z,t) & n\notin\mathcal Z\\
	    \mathsf f(\mathcal Z-\{n\},t+1) & n\in\mathcal Z
	    \end{cases}
	\end{equation}}
	\begin{proof}
	    It is trivial when $n\notin\mathcal Z$. So, suppose $n\in\mathcal Z$. Then,
	    \begin{align*}
	        f'
	        &= \left.\left(\prod_{\beta\in\operatorname{span}\{x_i: i\in\mathcal Z\}}(x-\beta)^{q^t}\right)\right|_{x_n=0}\\
	        &= \left.\left(\prod_{\beta\in\operatorname{span}\{x_i: i\in\mathcal Z-\{n\}\}}\:\prod_{\alpha\in\mathbb F_q}(x-\beta-\alpha x_n)^{q^t}\right)\right|_{x_n=0}\\
	        &= \prod_{\beta\in\operatorname{span}\{x_i: i\in\mathcal Z-\{n\}\}}\:\prod_{\alpha\in\mathbb F_q}(x-\beta)^{q^t}\\
	        &= \prod_{\beta\in\operatorname{span}\{x_i: i\in\mathcal Z-\{n\}\}}(x-\beta)^{q^{t+1}}\\
	        &= \mathsf f(\mathcal Z-\{n\},t+1)\in\mathcal L_{n-1,k}
	    \end{align*}
	\end{proof}
\end{itemize}

\section{Generalized Hall's Theorem}
\label{appx:hallsthm}
Let $G=(U,V,E)$ represent the bipartite graph with the bipartite sets of vertices $U$ and $V$ and the edges $E\subset U\times V$. Let $N_G(\Omega)\subset V$ denote the neighborhood of $\Omega\subset U$, i.e. the set of all vertices in $V$ adjacent to some element of $\Omega$.
\begin{theorem}[Generalized Hall's Theorem]
    Let $G=(U,V,E)$ be a bipartite graph. Suppose that there exist integers $c\geq 0$ and $d_i\geq 1$ for $i\in U$ such that for any nonempty $\Omega\subset U$,
    \begin{equation}
        \label{eq:halls_ineq}
        |N_G(\Omega)|\geq c+\sum_{i\in\Omega}d_i
    \end{equation}
    Then, one can keep removing edges from $E$ without violating any of the inequalities until the degree of $i$ is exactly $c+d_i$ for all $i\in U$. \hfill$\diamond$
\end{theorem}
\begin{proof}
    We will do induction on $|U|$. If $|U|=1$, it is trivial. Let $n\geq 2$ and suppose it is true when $|U|<n$.
    Let $|U|=n$. We consider two cases:
    \begin{enumerate}
        \item (\ref{eq:halls_ineq}) is \emph{tight} for some $\Omega$ with $2\leq |\Omega|\leq n-1$.
        
        Let $G_1=(\Omega, V, E_1)$, where $E_1=E\cap (\Omega\times V)$ and
        $G_2=( \Omega^c\cup\{\Omega\},V,E_2)$,
        where $\Omega^c=U-\Omega$ and
        $$E_2=(E-E_1)\cup \{(\Omega,j):j\in N_G(\Omega)\}$$
        In other words, to obtain $G_2$, we merge the vertices in $\Omega$ into a single vertex called $\Omega$ with the edges from that to every vertex in $N_G(\Omega)$.
        Furthermore, let $d_{\Omega}=\sum_{i\in\Omega}d_i$.
        
        We will show that (\ref{eq:halls_ineq}) holds for $G_1$ and $G_2$ if and only if it holds for $G$. ($\Leftarrow$ direction is trivial)
        Let $\Omega_1\subset\Omega, \Omega_2\subset \Omega^c$. Then,
        \begin{align*}
            |N_G(\Omega_1\cup\Omega_2)|
            &=|N_G(\Omega_1)|+|N_G(\Omega_2)-N_G(\Omega_1)|\\
            &\geq |N_G(\Omega_1)| + |N_G(\Omega_2)-N_G(\Omega)|\\
            &= |N_G(\Omega_1)| + |N_G(\Omega_2\cup\Omega)|-|N_G(\Omega)|\\
            &= |N_{G_1}(\Omega_1)| + |N_{G_2}(\Omega_2\cup\{\Omega\})|-(c+d_{\Omega})\\
            &\geq \left(c+\sum_{i\in\Omega_1}d_i\right) + \left(c+d_{\Omega}+\sum_{i\in\Omega_2}d_i\right)-(c+d_{\Omega})\\
            &= c+\sum_{i\in\Omega_1\cup\Omega_2}d_i
        \end{align*}
        
        Since $|\Omega|\leq n-1$ and $|\Omega^c\cup\{\Omega\}|\leq n-1$, by the induction hypothesis, we can remove edges from $G_1$ and $G_2$ until the degree of $i$ is $c+d_i$ for all $i\in U$. (Note that none of the edges from the vertex $\Omega$ in $G_2$ will be removed since its degree is already $c+d_{\Omega}$.)
        
        \item (\ref{eq:halls_ineq}) is \emph{strict} for all $\Omega$ with $2\leq |\Omega|\leq n-1$.\\
        If there exists an edge $(i,j)\in E$ such that the degree of $i$ is at least $c+d_i+1$ and the degree of $j$ is at least $2$, then removing $(i,j)$ will not violate (\ref{eq:halls_ineq}) because all the inequalities are strict except for $|\Omega|=n$, in which case, the left hand side is not affected.
        Now, we can assume that if a vertex $i\in U$ has degree at least $c+d_i+1$, then it is disconnected from the other vertices in $U$. Then, removing any edge from such a vertex $i$ will not violate any of the inequalities.\qedhere
    \end{enumerate}
\end{proof}

As a special case, letting $c=0$ and $d_i=1$ for all $i$ yields to the Hall's Marriage Theorem:
\begin{corollary}[Hall's Theorem]
    Let $G=(U,V,E)$ be a bipartite graph. If $|N_G(\Omega)|\geq |\Omega|$ for all $\Omega\subset U$, then there is a one-to-one matching from $U$ to $V$.\hfill$\diamond$
\end{corollary}

Letting $c=|V|-|U|$ and $d_i=1$ for all $i$ yields to the following corollary, which is also proved in \cite[Theorem 2]{dau2014existence}:

\begin{corollary}
    Let $\mathcal Z_1,\mathcal Z_2,\dots,\mathcal Z_k\subset[n]$ such that for all nonempty $\Omega\subset[k]$,
    \begin{equation}
        \left|\bigcap_{i\in\Omega}\mathcal Z_i\right|+|\Omega|\leq k
    \end{equation}
    Then, one can keep adding elements from $[n]$ to these subsets without violating any of the inequalities until each subset has exactly $k-1$ elements.\hfill$\diamond$
\end{corollary}
\begin{proof}
    Consider the bipartite graph $G=([k],[n],E)$ where $E=\{(i,j): j\notin\mathcal Z_i\}$. Then, $\left|\bigcap_{i\in\Omega}\mathcal Z_i\right|=n-|N_G(\Omega)|$ and the inequality becomes
    \begin{equation}
        |N_G(\Omega)|\geq (n-k) + |\Omega|
    \end{equation}
    Note that removing edges from the graph corresponds to adding elements from $[n]$ to the subsets.
\end{proof}

\end{document}